\documentclass[11pt]{article}
\usepackage{times}
\usepackage[english]{babel}
\usepackage[dvips]{graphicx}
\usepackage{amssymb}
\usepackage{amsfonts}
\usepackage{amsmath}
\usepackage{mathrsfs}
\usepackage{url} 

\newcommand{\handout}[5]{
   \renewcommand{\thepage}{#1-\arabic{page}}
   \noindent
   \begin{center}
   \framebox{
      \vbox{
    \hbox to 5.78in {{\sf CS 282A/MATH 209A: Foundations of Cryptography} \hfill \sf #2 }
       \vspace{4mm}
       \hbox to 5.78in { {\Large \hfill #5  \hfill} }
       \vspace{2mm}
       \hbox to 5.78in { {\em #3 \hfill #4} }
      }
   }
   \end{center}
   \vspace*{4mm}
}

\newtheorem{theorem}{Theorem}

\newtheorem{lemma}[theorem]{Lemma}

\newtheorem{claim}[theorem]{Claim}

\newcommand{\qed}{\rule{7pt}{7pt}}

\newenvironment{proof}{\noindent{\bf Proof}\hspace*{1em}}{\qed\bigskip}
\newenvironment{proof-sketch}{\noindent{\bf Sketch of Proof}\hspace*{1em}}{\qed\bigskip}
\newenvironment{proof-idea}{\noindent{\bf Proof Idea}\hspace*{1em}}{\qed\bigskip}
\newenvironment{proof-of-lemma}[1]{\noindent{\bf Proof of Lemma #1}\hspace*{1em}}{\qed\bigskip}
\newenvironment{proof-attempt}{\noindent{\bf Proof Attempt}\hspace*{1em}}{\qed\bigskip}

\makeatletter
\def\fnum@figure{{\bf Figure \thefigure}}
\def\fnum@table{{\bf Table \thetable}}
\long\def\@mycaption#1[#2]#3{\addcontentsline{\csname
  ext@#1\endcsname}{#1}{\protect\numberline{\csname
  the#1\endcsname}{\ignorespaces #2}}\par
  \begingroup
    \@parboxrestore
    \small
    \@makecaption{\csname fnum@#1\endcsname}{\ignorespaces #3}\par
  \endgroup}
\def\mycaption{\refstepcounter\@captype \@dblarg{\@mycaption\@captype}}
\makeatother

\newcommand{\mathify}[1]{\ifmmode{#1}\else\mbox{$#1$}\fi}

\newcommand{\bigO}O
\newcommand{\set}[1]{\mathify{\left\{ #1 \right\}}}

\usepackage{psbox, amssymb}

\newcommand{\al}{{\rm AL}}

\def\pr{\mathbb{P}}
\def\E{\mathbb{E}}

\def\bern{\textrm{Bernoulli}}
\def\bin{{\rm Bin}}

\def\ind{\mathbb{I}}

\def\Hcall{\mathcal{H}}

\def\Ncal{\mathcal{N}}

\def\al{\alpha}
\def\bt{\beta}
\def\dl{\delta}
\def\lm{\lambda}

\def\eps{\epsilon}

\def\Lm{\Lambda}
\def\lm{\lambda}

\def\zl{\zeta_{\lm}}
\def\o{o}

\def\beq{\begin{equation}}
\def\eeq{\end{equation}}

\def\beqn{\begin{eqnarray}}
\def\eeqn{\end{eqnarray}}

\newif\ifnotesw\noteswtrue
\newcommand{\comment}[1]{\ifnotesw $\blacktriangleright$\ {\sf #1}\ 
  $\blacktriangleleft$ \fi}

\def\whp{\textbf{whp}}

\title{Component Evolution in General Random Intersection Graphs}
\author{Milan Bradonji\'c\thanks{Theoretical Division, and Center for Nonlinear Studies, Los Alamos National Laboratory, Los Alamos, NM 87545, USA,~\texttt{milan@lanl.gov},}, 
Aric Hagberg\thanks{Theoretical Division, Los Alamos National Laboratory, Los Alamos, NM 87545, USA,~\texttt{hagberg@lanl.gov},},
Nicolas W. Hengartner\thanks{Statistical Sciences Group 
Los Alamos National Laboratory, NM 87545, USA,~\texttt{nickh@lanl.gov},}, 
Allon G. Percus\thanks{School of Mathematical Sciences
Claremont Graduate University, Claremont, CA 91711, USA,~\texttt{allon.percus@cgu.edu}.
}
}

\date{}

\begin{document}

\maketitle

\begin{abstract}
Random intersection graphs (RIGs) are an important random structure with
applications in social networks, epidemic networks, blog readership, and
wireless sensor networks.  RIGs can be interpreted as a model for large
randomly formed non-metric data sets.
We analyze the component evolution in general RIGs,
and give conditions on existence and uniqueness of the giant component. 
Our techniques generalize existing methods for analysis of component
evolution:
we analyze survival and extinction properties of a dependent,
inhomogeneous Galton-Watson branching process on general RIGs.
Our analysis relies on bounding the branching processes and inherits the
fundamental concepts of the study of component evolution in
Erd\H{o}s-R\'enyi graphs.
%
%
The major challenge comes from the underlying structure of RIGs, 
which involves its both the set of nodes and the set of attributes, 
as well as the set of different probabilities among the nodes and attributes.

\vspace{0.2cm}
\end{abstract}

\section{Introduction}
\label{sec.introduction}
%

Bipartite graphs, consisting of two sets of nodes with edges only
connecting nodes in opposite sets, are a natural representation for many
networks.  A well-known example is a collaboration graph, where
the two sets might be scientists and research papers, or actors and
movies~\cite{watts-1998-collective,newman-2001-scientific}.  Social
networks can often be cast as bipartite graphs since they are
built from sets of individuals connected to sets of attributes, such as
membership of a club or organization, work colleagues, or fans of the
same sports team.  Simulations of epidemic spread in
human populations are often performed on networks constructed from
bipartite graphs of people and the locations they visit during a
typical day~\cite{eubank-2004-modelling}.  
Bipartite structure, of course, is hardly limited to social networks.
The relation between nodes and keys in secure wireless communication,
for examples, forms a bipartite network~\cite{bloznelis-2009-component}.
In general, bipartite graphs are well suited to the problem of
classifying objects,
where each object has a set of properties~\cite{godehardt-2007-random}.
However, modeling such classification networks remains a challenge.
The well-studied Erd\H{o}s-R\'enyi model, $G_{n,p}$,
successfully used for average-case analysis of algorithm performance,
does not satisfactorily represent 
many randomly formed social or collaboration networks.
For example, $G_{n,p}$ does not
capture the typical scale-free degree distribution of many real-world
networks~\cite{barabasi-1999-emergence}.  
More realistic degree
distributions can be achieved by the configuration
model~\cite{newman-2001-random} or expected degree
model~\cite{chung-2002-average}, but even those fail to capture common
properties of social networks such as the high number of triangles (or
cliques) and strong degree-degree
correlation~\cite{newman-2003-why,albert-2002-statistical}.  

The most straightforward way of remedying these problems is to
characterize each of the bipartite sets separately.  One step in this
direction is an extension of the configuration model that specifies
degrees in both sets~\cite{guillaume-2006-bipartite}.
Another related approach is that of random intersection graphs (RIG),
first introduced in~\cite{singer-1995-thesis,karonski-1999-random}.
Any undirected graph can be represented as
an intersection graph~\cite{erdos-1966-representation}.
The simplest version is the ``uniform'' RIG, $G(n,m,p)$, containing a
set of $n$ nodes and a set of $m$ attributes, where any given node-attribute
pair contains an edge with a fixed probability $p$, independently of
other pairs.  Two nodes in the graph are taken to be connected if
and only if they are both connected to at least one common element in
the attribute set.  In our work, we study the more general RIG,
$G(n,m,\boldsymbol{p})$
\cite{nikoletseas-2004-existence,nikoletseas-2008-large}, where
the node-attribute edge probabilities are not given by a uniform
value $p$ but rather by a set $\boldsymbol{p} = \{p_w\}_{w \in W}$: a node is attached to the 
attribute $w$, with probability $p_w$.
This general model has only recently been developed and only a few
results have obtained, such as expander properties, cover time, and  
the existence and efficient construction of large independent
sets~\cite{nikoletseas-2004-existence,nikoletseas-2008-large,spirakis-2009-expander}.

In this paper, we analyze the evolution of components in general RIGs.
Related results have previously been obtained for the uniform
RIG~\cite{behrisch-2007-component},
and for two uniform cases of the RIG model where a specific overlap threshold controls the
connectivity of the nodes, were analyzed in~\cite{bloznelis-2009-component}.
%
Our main contribution is a generalization of the component evolution
on a general RIG.
We provide stochastic bounds, by analyzing the stopping time of the branching process on general RIG, 
where the history of the process is directly dictated by the structure of the general RIG.
%
%
The major challenge comes from the underlying structure of RIGs, 
which involves both the set of nodes and the set of attributes, 
as well as the set of different probabilities $\boldsymbol{p} = \{p_w\}_{w \in W}$.
%
%
%
%
%
%
%

\section{Model and previous work}
%
%
In this paper, we will consider the general intersection graph $G(n,m,\boldsymbol{p})$, introduced in~\cite{nikoletseas-2004-existence,nikoletseas-2008-large},  
with a set of probabilities $\boldsymbol{p} = \{p_w\}_{w \in W}$, where $p_w\in (0,1)$.
We now formally define the model. 

\textbf{Model.} 
There are two sets: the set of
nodes $V = \{1,2,\dots,n\}$ and the set of attributes $W =
\{1,2,\dots,m\}$. 
For a given set of probabilities $\boldsymbol{p} = \{p_w\}_{w \in W}$, independently over all $(v,w) \in V \times W$ let 
\beq
\label{eq.Avw}
A_{v,w} := \bern(p_w).
\eeq
Every node $v \in V$ is assigned a random set of attributes $W(v) \subseteq W$
\beq
\label{eq.Wv}
W(v) := \{w \subseteq W \mid A_{v,w} = 1\}.
\eeq
The set of edges in $V$ is defined such that two different nodes $v_i,v_j \in V$ are connected if and
only if 
\beq
\label{eq.connectivity}
|W(v_i) \cap W(v_j)| \geq s,
\eeq
for a given integer $s \geq 1$.
%
%
%

In our analysis, $p_w$ are not necessarily the same as in~\cite{behrisch-2007-component, bloznelis-2009-component}~\footnote{Note that $p_w$'s do not sum up to $1$. Moreover, we can eliminate the cases $p_w = 0$ and $p_w = 1$. These two cases respectively correspond when none or all nodes $v$ are attached to the attribute $w$.}, 
and for simplicity we fix $s=1$.

%
%

The component evolution of the uniform model $G(n,m,p)$ was
analyzed by Behrisch in~\cite{behrisch-2007-component}, for the case when the
scaling of nodes and attributes is $m=n^{\alpha}$, with $\alpha \ne 1$ 
and $p^2 m = c/n$. Theorem 1 in~\cite{behrisch-2007-component} states 
that the size of the largest component $\Ncal(G(n,m,p))$ in RIG satisfies 
(i) $\Ncal(G(n,m,p)) \leq \frac{9}{(1-c^2)}\log n$, for $\al>1, c<1$, 
(ii) $\Ncal(G(n,m,p)) = (1+o(1))(1-\rho)n$, for $\al>1, c>1$,   
(iii) $\Ncal(G(n,m,p)) \leq \frac{10 \sqrt{c}}{(1-c^2)}\sqrt{\frac{n}{m}} \log m$, for $\al<1, c<1$,
(iv) $\Ncal(G(n,m,p)) =(1+o(1))(1-\rho) \sqrt{c m n}$, for $\al<1, c>1$,   
where $\rho$ is the solution in $(0,1)$ of the equation $\rho = \exp(c(\rho-1))$.

%
The component evolution for the case $s \geq 1$ in the relation 
 $|W(u) \cap W(v)| \geq s$ is considered in~\cite{bloznelis-2009-component}, where the following two RIG models are analyzed:
(1) $G_s(n,m,d)$ model, where $\pr[W(v) = A] = {m \choose d }^{-1}$ for all $A \subseteq W$ on $d$ elements, for a given $d$;
(2) $G'_s(n,m,p)$ model, where $\pr[W(v) = A] = p^{|A|}(1-p)^{m-|A|}$ for all $A \subseteq W$.
In light of results of~\cite{behrisch-2007-component}, it has been shown in~\cite{bloznelis-2009-component}, that for 
$d = d(n), p=p(n), m=m(n), n = o(m)$, where $s$ is a fixed integer, and $d^{2s} \sim c m^s s!/n$,
the largest component in $G_s(n,m,d)$ satisfies: 
(i)  $\Ncal(G_{s}(n,m,d))  \leq \frac{9}{(1-c^2)}\log n$, for $c<1$,
(ii) $\Ncal(G_{s}(n,m,d))   = (1+o(1))(1-\rho)n$, for $c>1$, 
in the case when $n \log n = o(m)$ for $s=1$ and $n = o(m^{s/(2s-1)})$ for $s\leq 2$.
The same results for the giant component in $G_s(n,m,p)$ still hold for the case when 
$p^{2s} = cs! / m^s n$ and $n = o(m^{s/(2s-1)})$, see~\cite{bloznelis-2009-component}.
%
 
Both $G_s(n,m,d)$ and $G'_s(n,m,p)$
are special cases of a more general class studied in~\cite{godehardt-2001-two},
where the number of attributes of each node is assigned
randomly as in the bipartite configuration model.
That is, for a given probability distribution $(P_0, P_1, \dots, P_m)$, 
we have $\pr[|W(v)| = k] = P_k$ for all $0 \leq k \leq m$, and moreover given the size $k$, all of the sets $W(v)$ are equally probable,
that is for any $A \subseteq W$, $\pr[W(v) = A : |W(v)| = k] = {m \choose k}^{-1}$.
That is, we see that $G_s(n,m,d)$ is equivalent to the model of~\cite{godehardt-2001-two}
with the delta-distribution, where the probability of the
$d$-th coordinate is $1$, while $G'_s(n,m,d)$ is equivalent to the
model of~\cite{godehardt-2001-two} with the $\bin(m,p)$ distribution.
%
%
To complete the picture of previous work, in~\cite{deijfen-2009-random}, it was shown that when $n=m$ a set of probabilities $\boldsymbol{p} = \{p_w\}_{w \in W}$ can be chosen to tune the degree and clustering coefficient of the graph.

\section{Mathematical preliminaries}
\label{sec.math.prelim.}
In this paper, we analyze the component evolution of the general RIG structure.
As we have already mentioned, the major challenge comes from the underlying structure of RIGs, 
which involves both the set of nodes and the set of attributes, 
as well as the set of different probabilities $\boldsymbol{p} = \{p_w\}_{w \in W}$.

Moreover, the edges in RIG are not independent. Hence, a RIG cannot be treated as an Erd\H{o}s-R\'enyi random graph $G_{n,\hat p}$, with the edge probability $\hat p = 1 - \prod_{w \in W} (1-p_w^2)$.  
However, in~\cite{fill-2000}, the authors provide the comparison among $G_{n,\hat p}$ and $G(n,m,p)$, showing that for $m = n^\alpha$ and $\al > 6$, 
these two classes of graphs have asymptotically the same properties. 
In~\cite{rybarczyk-2009}, Rybarczyk has recently shown the equivalence of sharp threshold functions among $G_{n,\hat p}$ and $G_{n,m,p}$, when $m \geq n^3.$
In this work, we do not impose any constraints among $n$ and $m$, and we develop methods for the analysis of branching processes on RIGs, 
since the existing methods for the analysis of branching processes on $G_{n,p}$ do not apply.

We now briefly state the edge dependence. Consider  three distinct nodes $v_i,v_j,v_k$ from $V$.  
Conditionally on the set $W(v_k)$, by the definition~(\ref{eq.Wv}),
the sets $W(v_i) \cap W(v_k)$ and $W(v_j) \cap W(v_k)$ are mutually independent,
which implies conditional independence of
the events $\{v_i \sim v_k \mid W(v_k)\}, \{v_j \sim v_k \mid W(v_k)\}$, that is, 
\beq
\label{eq:two-nodes-ind}
\pr [ v_i \sim v_k , v_j \sim v_k \mid W(v_k) ]  = \pr [ v_i \sim v_k \mid W(v_k) ] \pr [ v_j \sim v_k \mid W(v_k) ].
\eeq
However, the latter does not imply independence of the
events $\{v_i \sim v_k\}$ and $\{v_j \sim v_k\}$ since in general
\begin{eqnarray}
\label{eq.rig.depend.}
\nonumber \pr [ v_i\sim v_k , v_j \sim v_k  ] &=& {\mathbb E}[\pr [ v_i \sim v_k , v_j \sim v_k \mid W(v_k) ] \\
\nonumber &=& {\mathbb E} \left [\pr [ v_i \sim v_k \mid W(v_k) ] \pr [ v_j \sim v_k \mid W(v_k) ] \right ]\\
&\neq& {\mathbb P}[v_i \sim v_k]{\mathbb P}[v_j \sim v_k].
\end{eqnarray}
Furthermore, the conditional pairwise independence (\ref{eq:two-nodes-ind}) 
does not extend to 
three or more nodes. Indeed, conditionally on the set $W(v_k)$, the sets  
$W(v_i) \cap W(v_j), W(v_i) \cap W(v_k)$, and $W(v_j) \cap W(v_k)$ are not mutually independent,
and hence neither are the events $\{v_i \sim v_j\}, \{v_i\sim v_k\}$, and  $\{v_j \sim v_k\}$,
that is,
\beq
\label{eg:tree-nodes-dep}
\pr[ v_i\sim v_j, v_i \sim v_k, v_j \sim v_k \mid W(v_k) ] \neq \pr[ v_i \sim v_j \mid W(v_k) ] \pr[ v_i \sim v_k \mid W(v_k)] \pr[ v_j \sim v_k \mid W(v_k)].
\eeq

We now provide two identities, which we will use  throughout this paper.
For any $w \in W$, let $q_{w} := 1 - p_{w}$, and define
$\prod_{\al \in \emptyset } q_{\al}  = 1$. 

\begin{claim} \label{claim:1}
For any node $u \in V$ and given set  $A \subseteq W$,
\begin{equation} 
\label{eq:prob.A.empty}
\pr[W(u) \cap A = \emptyset | A ] = \prod_{\al \in A} (1-p_{\al}) = \prod_{\al \in A} q_{\al}.
\end{equation}
\end{claim}
\begin{proof}
Write
\beq
\nonumber
\pr[W(u) \cap A = \emptyset | A ] = \pr[\forall \alpha \in A , \alpha \notin W(u) | A] =  
\prod_{\al \in A} \pr[\alpha \notin W(u) ]
= \prod_{\al \in A} (1-p_{\al}) = \prod_{\al \in A} q_{\al},
\eeq
which is the desired expression.
\end{proof}

\begin{claim}
\label{cl:AB}
For any node $u \in V$, and given sets $A \subseteq B \subseteq W$,
\begin{equation}
\pr[W(u) \cap A = \emptyset, W(u) \cap B \neq \emptyset | A, B] =
 \Big( \prod_{\al \in A} q_{\al} \Big) \Big( 1 - \prod_{\al \in B \setminus A} q_{\beta} \Big) \nonumber = \prod_{\al \in A} q_{\al} - \prod_{\bt \in B} q_{\bt}.
\end{equation}
\end{claim}
\begin{proof}
The sets $A$ and $B \setminus A$ are disjoint.  The result follows from (\ref{eq:prob.A.empty}).
\end{proof}

\section{Auxiliary process on general random intersection graphs}
\label{sec:branching.process}
Our analysis for the emergence of a giant component is inspired by the
approach described in~\cite{alon-2000-probabilistic}. 
%
The difficulty in analyzing the evolution of the 
stochastic process defined by equations~(\ref{eq.Avw}),~(\ref{eq.Wv}), 
and~(\ref{eq.connectivity}) resides in the fact that we need, at least in 
principle, to keep track of the temporal evolution of the sets of nodes and 
attributes being explored. This results in a process that is not Markovian.  

We construct an auxiliary process, which starts at an arbitrary node $v_0 \in V$, 
and reaches zero for the first time in a number of steps equal to the size of the
component containing $v_0$.  The process is algorithmically defined as follows.


\textbf{Auxiliary Process.}
Let us denote by $V_t$ the cumulative set of nodes \textit{visited} by time $t$, which we initialize
to $V_0=\{v_0\}$, and set $W(v_0) = \{ v \not = v_0 : W(v) \cap W(v_0) \not = \emptyset \}$.  
Starting with $Y_0=1$, the process evolves as follows:  For $t=1,2,3,\dots,n-1$ and $Y_t>0$,
pick a node $v_t$ uniformly at random from the set  $V \setminus V_{t-1}$ and update
the set of visited nodes $V_t = V_{t-1} \cup \{v_t\}$.  Denote by 
 $W(v_t) = \{w \in W \mid A_{v_t,w} = 1\}$ the set of features associated to node
$v_t$, and define
\[
Y_t = \left | \Big\{ v \in V \setminus V_t \mid W(v) \cap \cup_{\tau=0}^t W(v_\tau)\neq \emptyset \Big\}
\right |.
\]
The random variable $Y_t$ counts the number of nodes outside the set of 
visited nodes $V_t$ that are connected to $V_t$.   Following
\cite{alon-2000-probabilistic}, we call $Y_t$ the number of \textit{alive} nodes
at time $t$.  We note that we do not need to keep track of the actual list of neighbors of  $V_t$
\begin{equation} \label{eq:neighbor-set}
\Big\{ v \in V \setminus V_t \mid W(v) \cap \cup_{\tau=0}^t W(v_\tau)\neq \emptyset \Big\},
\end{equation}
as in~\cite{alon-2000-probabilistic}, because every node in $V \setminus V_t$ is equally likely 
to belong to the set (\ref{eq:neighbor-set}).  As a result, each time we need a random node
from (\ref{eq:neighbor-set}),  we pick a node uniformly at random form $V \setminus V_t$.

To understand why this process is useful, notice that by time $t$, we know that the size
of the component containing $v_0$ is at least as large as the number of visited nodes 
$V_t$ plus the number $Y_t$ of neighbors of $V_t$ not yet visited.  Once the number $Y_t$
of neighbors connected to $V_t$ but not yet visited drops to zero, the size of $V_t$ 
is equal to the size of the component containing $v_0$.  We formalize this last statement
by introducing the stopping time 
\begin{equation} \label{eq:stopping.time}
T(v_0) = \inf \{ t > 0: Y_t = 0 \},
\end{equation}
whose value is $|C(v_0)|$.

Finally, our analysis of that process requires us to keep track of the history of the
feature sets uncovered by the process
\beq
\label{eq:history}
\Hcall_t = \{W(v_0), W(v_1), \dots, W(v_t)\}.
\eeq

\subsection{Process description in terms of random variable $Y_t$}
As in~\cite{bloznelis-2009-component}, we denote the cumulative 
feature set associated to the sequence of 
nodes $v_0, \dots, v_t$ from the auxiliary process by
\beq
W_{[t]} := \cup_{\tau=0}^{t} W(v_\tau). 
\eeq
We will characterize the process $\{Y_t\}_{t \geq 0}$ in terms of the number $Z_t$ of newly 
discovered neighbors to $V_t$.  The latter is directly related to the increment, defined by
of the process $Y_t$
\beq
\label{eq.increment.yt}
Z_t = Y_t - Y_{t-1} + 1,
\eeq
where the term +1 reflects the fact that one node,  $Y_{t-1}$ decreases by one when the 
node $v_t$ becomes a visited node at time $t$.  The events that any given node, 
which is neither visited nor alive, becomes
\textit{alive} at time $t$ are conditionally independent given the history 
${\mathcal H}_t$, since each event involves a different subsets of  the 
indicator random variables $\{ A_{v,w} \}$.   In light of Claim~\ref{cl:AB},
the conditional probability that a node $u$ becomes  alive at time $t$ is
\begin{eqnarray}
\label{eq:rate}
\nonumber r_t &:=& \pr[u \sim v_t, u \not\sim v_{t-1}, u \not\sim v_{t-2}, \dots, u \not\sim v_0 | \Hcall_t ] \\
\nonumber &=& \pr[W(u) \cap W(v_t) \neq \emptyset, W(u) \cap W_{[t-1]} = \emptyset | \Hcall_t ]\\ 
\nonumber &=& \pr[W(u) \cap W(v_t) \neq \emptyset, W(u) \cap W_{[t-1]} = \emptyset | W(v_t), W_{[t-1]} ]\\ 
\nonumber &=& \prod_{\al \in W_{[t-1]}} q_{\al} - \prod_{\bt \in W_{[t]}} q_{\bt}\\
	  &=& \phi_{t-1} - \phi_{t},
\end{eqnarray}
where we set $\phi_{t} := \prod_{\al \in W_{[t]}} q_{\al}$,
and use the convention  $W_{[-1]} = W(\emptyset) \equiv \emptyset$ 
and $\phi_{-1} \equiv 1$.   Observe that the probability (\ref{eq:rate}) 
does not depend on $u$.    Hence the 
number of new alive nodes at time $t$ is, conditionally on the history $\Hcall_t $,
a Binomial distributed random variable with parameters $r_t$
and 
\beq
\label{eq:nt} 
N_t = n - t - Y_{t}.
\eeq
Formally,
\begin{equation} \label{eq:zt+1}
Z_{t+1} | \Hcall_t \sim \mbox{Bin}(N_t,r_t).
\end{equation}
This allows us to describe the distribution of $Y_t$ in the next lemma.

\begin{lemma}
\label{lm:nt}
For times $t \geq 1$, the number of alive nodes satisfies
\begin{equation}
Y_t  | \Hcall_{t-1} \sim  \bin \Big(n-1, 1-\prod_{\tau = 0}^{t-1}(1-r_\tau) \Big) -t +1.
\end{equation} 
\end{lemma}
The proof of this lemma requires us to establish the following result first.
\begin{lemma}
\label{lm:binbin}
Let random variables $\Lm_1, \Lm_2$ satisfy: $\Lm_1 \sim \bin(m,\nu_1)$ and $\Lm_2 \textrm{ given } \Lm_1 \sim \bin(\Lm_1,\nu_2)$.
Then marginally $\Lm_2 \sim \bin(m, \nu_1\nu_2)$ and 
$\Lm_1 - \Lm_2 \sim \bin(m, \nu_1 (1-\nu_2))$.
\end{lemma}
\begin{proof}
Let $U_1,\ldots,U_m$ and $V_1,\ldots,V_m$ be i.i.d. Uniform$(0,1)$ random variables.
Writing
\beq
\nonumber
\Lm_1 \stackrel{d}{=} \sum_{j=1}^m {\mathbb I}(U_j \leq \nu_1) \quad
\mbox{and}  \quad 
\Lm_2 | \Lm_1  \stackrel{d}{=}  \sum_{k : U_k < \nu_1} {\mathbb I}(V_k \leq \nu_2),
\eeq
we have that 
\beq
\nonumber
\Lm_2  \stackrel{d}{=}  \sum_{k=1}^m {\mathbb I}(U_k \leq \nu_1) {\mathbb I}(V_k \leq \nu_2)\\
\stackrel{d}{=} \sum_{k=1}^m {\mathbb I}(U_k \leq \nu_1 \nu_2),
\eeq
from which the conclusion follows.
\end{proof}

\begin{proof}(Proof of Lemma~\ref{lm:nt})
We prove the assertion on the Lemma by induction in $t$. 
For $t=0$, $Y_0 = 1$ and $t=1$,  $Y_1 = Z_1 ~\sim \bin(n-1,r_0)$. 
Hence, the Lemma is true for $t = 1$ and $t=0$.  
Assume that the assertion is true for some $t\geq 1$,
\begin{equation}
Y_t  | \Hcall_{t-1} \sim  \bin \Big(n-1, 1-\prod_{\tau = 0}^{t-1}(1-r_\tau) \Big) -t +1.
\end{equation}
From (\ref{eq:zt+1}), we have $Z_{t+1} | \Hcall_t \sim \bin(N_t,r_t) = \bin(n-t -Y_t,r_t)$, Now, from~(\ref{eq.increment.yt})
and Lemma~\ref{lm:binbin}, it follows 
\begin{equation}
Y_{t+1}  | \Hcall_{t} \sim  \bin \Big(n-1, 1-\prod_{\tau = 0}^{t}(1-r_\tau) \Big) - t.
\end{equation}
Hence, by mathematical induction, the Lemma holds for any $t \geq 0$.
\end{proof}

\subsection{Expectation and variance of $\phi_t$}

The history ${\mathcal H}_t$ embodies the evolution of how the features 
are discovered over time.   It is insightful to recast that history in terms
of the discovery times  $\Gamma_w$ of each feature in $W$.   Given any sequence
of nodes $v_0,v_1,v_2,\ldots$, the probability that a given feature $w$
is first discovered at time $t < n$ is
\begin{eqnarray*}
{\mathbb P}[\Gamma_w = t ] & = &
{\mathbb P}[A_{v_t,w}=1,A_{v_{t-1},w}=0,\ldots,A_{{v_0},w}=0] \\
&= & p_w (1-p_w)^t.
\end{eqnarray*}
If a feature $w$ is not discovered by time $n-1$, we set $\Gamma_w = \infty$ and note that
\beq
\nonumber
{\mathbb P}[\Gamma_w = \infty ] = (1-p_w)^n.
\eeq
From the independence of the random variables $A_{v,w}$, it follows that 
the discovery times $\{ \Gamma_w :  w\in W \}$ are independent.
We now focus on describing the distribution of $\phi_t = \prod_{\al \in W_{[t]}} q_{\al}$.
For $t \geq 0$, we have
\beq
\label{eq:phi.geom.}
\phi_t = \prod_{\al \in W_{[t]}} q_\al = \prod_{j=0}^t \prod_{\al \in s(v_j) \setminus S[j-1]} q_\al 
\stackrel{d}{=} \prod_{j=0}^t \prod_{w \in W} q_w^{{\mathbb I}(\Gamma_w=j)}
= \prod_{w \in W} q_w^{ \ind(\Gamma_w \leq t)}. 
\eeq
Using the fact that for a $B \sim \textrm{Bernoulli}(r)$, the expectation ${\mathbb E}[a^B] = 1-(1-a)r$, 
we can easily calculate the expectation of $\phi_t$
\begin{eqnarray}
\label{eq.phigeomexp}
\nonumber \E[\phi_t] &=& \E[ \prod_{w \in W} q_w^{ \ind(\Gamma_w \leq t)} ] = \prod_{w \in W} \Big(1- (1-q_w)\pr[\Gamma_w \leq t ]\Big)\\
&=&  \prod_{w \in W} \Big(1 -(1-q_w)(1-q_w^{t+1}) \Big).
\end{eqnarray}
The concentration of $\phi_0$ will be crucial for the analysis of the supercritical regime, Subsection~\ref{sub.supercritical.regime}. Hence, we here provide $\E[\phi_0]$ and $\E[\phi_0^2]$. 
From (\ref{eq.phigeomexp}) it follows
\beq
\label{eq.phi.zero.exp.}
\E[\phi_0]  = \prod_{w \in W} (1 - p_w^2) = 1 - \sum_{w \in W} p_w^2 + \o(\sum_{w \in W} p_w^2).
\eeq
Moreover, from (\ref{eq:phi.geom.}) it follows 
\begin{eqnarray}
\label{eq.phi.zero.sqaure.exp.}
\nonumber \E[\phi_0^2] &=& \E[ \prod_{w \in W} q_w^{ 2\ind(\Gamma_w \leq 0)} ] = \prod_{w \in W} \Big(1- (1-q_w^2)\pr[\Gamma_w = 0]\Big) = \prod_{w \in W} \Big(1- (1-q_w^2) p_w \Big) \\
&=& \prod_{w \in W} \Big(1- 2 p_w^2 + p_w^3 \Big) = 1 - 2 \sum_{w \in W} p_w^2  + \o(\sum_{w \in W} p_w^2).
\end{eqnarray}

\section{Giant component}
\label{sec:giant.component}
With the process $\{Y_t\}_{t \geq 0}$ defined in the previous section, we analyze
both the subcritical and supercritical regime of our random intersection graph
by adapting the percolation based techniques to analyze Erd\H{o}s-R\'enyi random graphs~\cite{alon-2000-probabilistic}.   
The technical difficulty in analyzing that stopping time rests in the fact that the 
distribution of $Y_t$ depends on the history of the process,
dictated by the structure of the general RIG.
In the next two subsections, we will give conditions on 
non-existence, that is, on existence and uniqueness of the giant component in general RIGs. 
%

\subsection{Subcritical regime}
\label{sub.subcritical.regime}
\begin{theorem}
\label{thm.sub.subcritical.regime}
Let 
\beq
\nonumber
\sum_{w \in W}  p_w^3 = O(1/n^2) \quad \mbox{ and } \quad p_w = O(1/n) \mbox{ for all } w. 
\eeq
For any positive constant $c<1$, if $\sum_{w \in W}  p_w^2 \leq c/n$,  then all components in a general random intersection graph $G(n,m,\boldsymbol{p})$ are of order $O(\log n)$, with high probability\footnote{We will use the notation ``with high probability'' and denote as \whp, meaning with probability $1- o(1)$, as the number of nodes $n \to \infty$.}.
\end{theorem}

\begin{proof}
We generalize the techniques used in the proof for the sub-critical case in $G_{n,p}$ presented in~\cite{alon-2000-probabilistic}.  Let $T(v_0)$ be the stopping time define in (\ref{eq:stopping.time}),
for the process starting at node $v_0$ and note that $T(v_0)=|C(v_0)|$.
We will bound the size of the largest component, and prove that under the conditions of the theorem, all  components are of order $O(\log n)$,  \whp.

For all $t \geq 0$,
\begin{eqnarray}
 \nonumber
{\mathbb P}[ T(v_0)>t  ] & = & {\mathbb E} \left [ \pr\left[T(v_0)>t \mid \Hcall_t \right]  \right ] \leq
 {\mathbb E} \left [  \pr[Y_t>0 \mid \Hcall_t ] \right ]\\ 
&=& {\mathbb E} \left [ \pr[\bin(n-1,1 - \prod_{\tau=0}^{t-1} (1-r_\tau)) \geq t \mid \Hcall_t ] \right ].
\label{eq:prob.stopping.time}
\end{eqnarray}
Bounding from above, which can easily be proven by induction in $t$ for $r_\tau \in [0,1]$, we have
\beq
\label{eq.rate.upper.bound}
1 - \prod_{\tau=0}^{t-1} (1-r_\tau) \leq \sum_{\tau=0}^{t-1} r_\tau = \sum_{\tau=0}^{t-1}
(\phi_{\tau-1}-\phi_\tau) = 1 - \phi_{t-1}.
\eeq
By using stochastic ordering of the Binomial distribution, both in $n$ and in $\sum_{\tau=0}^{t-1} r_\tau$, 
and for any positive constant $\nu$, which is to be specified later, it follows
\beqn
\label{eq:prtt}
\nonumber \pr[T(v_0) >t \mid \Hcall_t ] &\leq& \pr[\bin(n, \sum_{\tau=0}^{t-1} r_\tau) \geq t \mid \Hcall_t ] = \pr[\bin(n, 1 - \phi_{t-1}) \geq (1-\nu)t \mid \Hcall_t ] \\
\nonumber  &=&  \pr[\bin(n, 1 - \phi_{t-1}) \geq t \mid 1 - \phi_{t-1} < (1-\nu)t/n \cap \Hcall_t ] \pr[1 - \phi_{t-1} < (1-\nu)t/n \mid \Hcall_t ]\\
\nonumber  && +\:  \pr[\bin(n, 1 - \phi_{t-1}) \geq t \mid 1 - \phi_{t-1} \geq (1-\nu)t/n \cap \Hcall_t ] \pr[1 - \phi_{t-1} \geq (1-\nu)t/n \mid \Hcall_t ]\\
\nonumber &\leq & \pr[\bin(n, 1 - \phi_{t-1}) \geq t \mid 1 - \phi_{t-1} < (1-\nu)t/n \cap \Hcall_t ]  \\
&& + \pr[1 - \phi_{t-1} \geq (1-\nu) t/n \mid \Hcall_t ].
\eeqn
Furthermore, using the fact that the event $\{ 1 - \phi_{t-1} < (1 - \nu) t/n  \}$ is 
$\Hcall_t$-measurable, together with the stochastic ordering of the binomial distribution,
we obtain
\beq
\nonumber
\pr[\bin(n, 1 - \phi_{t-1}) \geq t \mid 1 - \phi_{t-1} <  (1-\nu) t/n \cap \Hcall_t ] 
\leq  \pr[\bin(n, (1-\nu) t/n) \geq t \mid \Hcall_t ].
\eeq
Taking the expectation with respect to the history ${\mathcal H}_t$ in (\ref{eq:prtt})
yields 
\[
{\mathbb P}[T(v_0) > t] \leq \pr[\bin(n, (1-\nu) t/n) \geq t] + \pr[1 - \phi_{t-1} \geq (1-\nu)t/n ].
\]
For $t = K_0 \log n$, where $K_0$ is a constant large enough and independent on the
initial node $v_0$, the Chernoff bound ensures that 
$\pr[\bin(n, (1-\nu) t/n) \geq t] = o(1/n)$.
%
To bound $\pr[1 - \phi_{t-1} \geq (1-\nu) t/n \mid \Hcall_t ]$, use (\ref{eq:phi.geom.}) to obtain
\begin{eqnarray*}
\{ 1 - \phi_{t-1} \geq (1-\nu) t/n \} &=&\left  \{ \prod_{w \in W} q_w^{ \ind(\Gamma_w \leq t)} \leq 1-\frac{(1-\nu) t}{n}
\right \}\\
&=& \left \{ \sum_{w \in W} \log \left ( \frac{1}{1-p_w} \right ) \ind(\Gamma_w \leq t) \geq -\log
\left ( 1-\frac{(1-\nu) t}{n} \right ) \right \}.
\end{eqnarray*}
Linearize $-\log(1-(1-\nu) t/n) = (1-\nu) t/n +o(t/n)$ and define the bounded auxiliary random variables
$X_{t,w} = n \log(1/(1-p_w)) \ind(\Gamma_w \leq t)$.  Direct calculations reveal that
\begin{eqnarray}
\nonumber \E[X_{t,w}] &=& n  \log \Big(\frac{1}{1-p_w} \Big) (1-q_w^t) 
= n \Big( p_w + \o(p_w) \Big) \Big(1-(1-p_w)^t \Big)\\
&=& n \Big( p_w + \o(p_w) \Big) \Big(tp_w + \o(t p_w)) \Big)
= n t   p_w^2 + \o \Big( n t  p_w^2\Big), 
\end{eqnarray}
which implies
\beq
 \sum_{w \in W}  \E[ X_{t,w}]  =  n t  \sum_{w \in W}  p_w^2 + \o \Big( n t \sum_{w \in W}  p_w^2\Big). 
\eeq
Thus under the stated condition that
\beq
\nonumber
n \sum_{w \in W} p_w^2 \leq c < 1,
\eeq
it follows that $0 < (1-c) t \leq t - \sum_{w \in W}  \E[ X_{t,w}]$.  In light of 
Bernstein's inequality~\cite{bernstein-1924}, we bound
\begin{eqnarray}
&& \pr[1 - \phi_{t-1} \geq (1-\nu) t/n] = {\mathbb P} \left [ \sum_{w \in W} X_{t,w} \geq (1-\nu) t \right ]
\leq  {\mathbb P} \left [ \sum_{w \in W} \big( X_{t,w}-{\mathbb E}[X_{t,w}] \big) \geq (1-\nu - c) t \right ] \nonumber \\
&\leq&
\exp\Big( -\frac{ \frac{3}{2} ((1-\nu - c) t )^2}{3 \sum_{w \in W} \mbox{Var}[X_{t,w}]  + n t\max_w \{p_w\} (1+\o(1)) } \Big). \label{eq:bernstein}
\end{eqnarray}
Since
\begin{eqnarray}
\nonumber \E[X_{t,w}^2] &=&  \Big(n  \log \Big(\frac{1}{1-p_w} \Big) \Big)^2 (1-q_w^t) = n^2 \Big( p_w + \o(p_w) \Big)^2 \Big(1-(1-p_w)^t \Big)\\
&=& n^2 \Big( p_w^2 + \o(p_w^2) \Big) \Big(tp_w + \o(t p_w)) \Big) =  n^2 t p_w^3 + \o \Big( n^2 t \sum_{w \in W}  p_w^3\Big), 
\end{eqnarray}
it follows that for some large constant $K_1>0$  
\beq
\nonumber
\sum_{w \in W}  \mbox{Var}[X_{t,w}] \leq \sum_{w \in W} \E[ X_{t,w}^2]  = n^2 t  \sum_{w \in W}  p_w^3 + \o \Big( n^2 t \sum_{w \in W}  p_w^3\Big) \leq K_1 t.
\eeq
Finally, the assumption of the theorem implies
that there exists constant $K_2 > 0$ such that
\beq
\nonumber
n \max_{w \in W} p_w \leq K_2.
\eeq
Substituting these bounds into (\ref{eq:bernstein}) yields
\beq
\nonumber
\pr[1 - \phi_{t-1} \geq (1-\nu) t/n] \leq \exp \left ( -\frac{3(1-\nu - c)^2}{2(3K_1+K_2)} t \right ),
\eeq
and taking $\nu \in (0,1-c)$ and $t=K_3 \log n$ for some constant $K_3$ large enough and not depending on the 
initial node $v_0$,  we conclude that
$\pr[1 - \phi_{t-1} \geq (1-\nu) t/n] = o(n^{-1})$, which in turn implies that taking constant $K_4 = \max \{ K_0,K_3\}$,
ensures that
\beq
\nonumber
\pr[T(v_0)> K_4 \log n] = \o(1/n)
\eeq
for any initial node $v_0$.
Finally, a union bound over the $n$ possible starting values $v_0$ implies that
\beq
\nonumber
{\mathbb P}[\max_{v_0 \in V} T(v_0) > K_4 \log n ] \leq n \o(n^{-1}) = o(1),
\eeq
which implies that all connected components in the random intersection 
are of size $O(\log n)$, \whp.
\end{proof}

\textbf{Remarks.} We now consider the conditions of the theorem.
From the Cauchy-Schwarz inequality, we obtain $\Big(\sum_{w \in W}  p_w^3\Big) \Big(\sum_{w \in W}  p_w\Big) \geq \Big(\sum_{w \in W}  p_w^2\Big)^2$. Moreover, given that $\sum_{w \in W}  p_w^3 = O(1/n^2)$ and $p_w = O(1/n)$, it follows $\sum_{w \in W}  p_w^2 = \Omega(\sqrt{m/n^3})$. Hence, for $\sum_{w \in W}  p_w^2 = c/n$, when $c<1$, it follows $m = \Omega(n)$, which is consistent with the results in~\cite{behrisch-2007-component} on the non-existence of a giant component in a uniform RIG.

\subsection{Supercritical regime}
\label{sub.supercritical.regime}

We now turn to the study of the supercritical regime in which 
$\lim_{n \rightarrow \infty} n \sum_{w \in W}  p_w^2 = c > 1$.  

\begin{theorem}
\label{thm.super.subcritical.regime}
Let 
\beq
\nonumber
\sum_{w \in W}  p_w^3 = o\Big(\frac{\log n}{n^2}\Big) \quad \mbox{ and } \quad p_w = o\Big(\frac{\log n}{n}\Big), \quad \mbox{ for all } w. 
\eeq
For any constant $c>1$, if 
$\sum_{w \in W}  p_w^2 \geq c/n$, then {\whp} 
there exists a unique largest component in $G(n,m,\boldsymbol{p})$, of order $\Theta(n)$. 
Moreover, the size of the giant component is given by $n \zeta_c (1 + \o(1))$,
where $\zeta_c$ is the solution in $(0,1)$ of the equation $1-e^{-c \zeta} = \zeta$,
while all other components are of size $O(\log n)$.
\end{theorem}

\noindent
\textbf{Remarks.}  The conditions on $p_w$ and $\sum_w p_w^3$ are weaker than 
ones in the case of the sub-critical regime. 

The proof proceeds as follows.   The first step is to bound, both from above and
below, the value  $1 - \prod_{\tau = 0}^{t-1}(1-r_\tau)$ that governs the behavior
the branching process  $\{Y_t\}_{t \geq 0}$, see Lemma~\ref{lm:nt}.  With the lower bound,
we show the emergence with high probability of at least one giant component 
of size $\Theta(n)$.  We use the upper bound to prove uniqueness of the giant component.
Technically, we make use of these bounds to compare our branching process to 
branching processes arising in the study of Erd\H{o}s-R\'eneyi random graphs.

\begin{proof}
We start by bounding $1 - \prod_{\tau = 0}^{t-1}(1-r_\tau)$.  The upper bounds 
$\sum_{\tau=0}^{t-1} r_{\tau}$ has been previously established in (\ref{eq.rate.upper.bound}).
For the lower bound, we apply Jensen's inequality to the function 
$\log (1-x)$ to get
\begin{eqnarray}
\nonumber \log \prod_{\tau = 0}^{t-1} (1-r_\tau) &=& \sum_{\tau = 0}^{t-1} \log (1-r_\tau)
= \sum_{\tau = 0}^{t-1} \log \Big( 1- (\phi_{\tau-1} - \phi_{\tau}) \Big) \\
&\leq& t \log\Big( 1- \frac{1}{t}\sum_{\tau=0}^{t-1}(\phi_{\tau-1} - \phi_{\tau})\Big) = t \log\Big( 1- \frac{1 - \phi_{t-1}}{t}\Big).
\end{eqnarray}
In light of (\ref{eq:phi.geom.}), $\phi_t$ is decreasing in $t$, and hence
\begin{equation}
\label{eq:rate.sandwich}
1- \Big( 1- \frac{1 - \phi_{0}}{t}\Big)^t \leq 1- \Big( 1- \frac{1 - \phi_{t-1}}{t}\Big)^t  \leq 1 - \prod_{\tau = 0}^{t-1}  (1-r_\tau) \leq \sum_{\tau=0}^{t-1} r_{\tau} = 1 - \phi_{t-1}.
\end{equation}
To further bound $1- \Big( 1- \frac{1 - \phi_{0}}{t}\Big)^t$, consider the function 
$f_t(x) = 1 - (1 - x/t)^t$ for $x$ in a neighborhood of the origin and $t \geq 1$.
For any fixed $x$, $f_t(x)$ decreases to $1-e^{-x}$ as $t$ tends to infinity.  The
latter function is concave, and hence for all $x \leq \varepsilon$, 
\beq
\nonumber
\frac{1-e^{-\varepsilon}}{\varepsilon} x \leq f_t(x).
\eeq
Note that $(1-e^{-\varepsilon})/\varepsilon$ can be made arbitrary close to one by taking 
$\varepsilon$ small enough.
Furthermore, $f_t(x)$ is increasing in $x$ for fixed $t$. 
From (\ref{eq:phi.geom.}), $1-\phi_0 \leq 1- \phi_t$, hence
$1- ( 1- \frac{1 - \phi_{0}}{t})^t \leq 1- ( 1- \frac{1 - \phi_{t-1}}{t})^t$. 
Looking closer at $1 - \phi_0$, from (\ref{eq.phi.zero.sqaure.exp.}) and (\ref{eq.phi.zero.exp.}), by using Chebyshev inequality, with $\sum_{w \in W} p_w^2 = c/n$, it follows that $\phi_0$, is concentrated around its mean $\E[\phi_0] = c/n$. That is, for any constant $\dl>0$, $\phi_0 \in ((1 - \delta)c/n, (1 + \dl)c/n)$, with probability $1 - o(1/n)$. 
We conclude that for any $\dl>0$ there is $\eps>0$ such that $(c - \dl)\frac{1 - e^{-\eps}}{\eps} > 1$, since constant $c>1$. Moreover, since $\lim_{\eps \to 0} \frac{1 - e^{-\eps}}{\eps} = 1$, by choosing $\eps$ sufficiently small, $\frac{1 - e^{-\eps}}{\eps}$ can be arbitrarily close to $1$.
It follows that $1 - \prod_{\tau = 0}^{t-1}  (1-r_\tau) > c'/n$, for some constant $c>c'>1$
arbitrarily close to $c$. 
Hence, the branching process on RIG is stochastically lower bounded by the $\bin(n-1, c'/n)$, which stochastically dominates a branching process on $G_{n,c'/n}$.  Because $c^\prime > 1$, 
there exists \whp \mbox{} a giant component of size $\Theta(n)$ in $G_{n,c'/n}$.   This implies
that the stopping of the branching process associated to $G_{n,c'/n}$ is $\Theta(n)$
with high probability, and so is the stopping time $T_{v}$ for some $v \in V$, 
which implies that there is a giant component in a general RIG, \whp. 

Let us look closer at the size of that giant component. 
From the representation (\ref{eq:phi.geom.}) for $\phi_{t-1}$, consider the previously introduced random variables $X_{t,w} = n \log(1/(1-p_w)) \ind(\Gamma_w \leq t)$.   
Similarly, as in the proof of the Theorem~\ref{thm.sub.subcritical.regime},
it follows that  under the conditions of the theorem
there is a positive constant $\dl>0$ such that $\sum_w X_{t,w}$ is concentrated within $(1\pm \delta)\sum_w\E[X_{t,w}] = 
(1 \pm \delta)c/n$, with probability $1 - o(1)$.
%
Hence, there exists $p^{+} = c^{+}/n$, for some 
constant $c^{+}>c>1$, such that $1 - \phi_{t-1} \leq 1 - (1 - p^{+})^t$, which is equivalent to $- \log \phi_{t-1} \leq t \log (1 - p^{+}) = tp^{+} + \o(tp^{+}) = t c^{+}/n + \o(t/n)$. 
Similarly, the concentration of $\phi_{t-1}$ implies that there  exists 
$p^{-} = c^{-}/n$, with $c > c^{-}>1$, such that $1 - (1 - p^{-})^t \leq 1- ( 1- (1 - \phi_{t-1})/t)^t$, 
which implies that $- \log \phi_{t-1} \geq t \log (1 - p^{-}) = tp^{-} + \o(tp^{-}) = t c^{-}/n + \o(t/n)$. 
Combining the upper and lower bound, we conclude that with probability $1 - o(1)$,
the rate of the branching process on RIG is bracketed by
\beq
\label{eq:rate.better.sandwich}
1 - (1 - p^{-})^t  \leq 1 - \prod_{\tau = 0}^{t-1}  (1-r_\tau) \leq 1 - (1 - p^{+})^t.
\eeq
The stochastic dominance of the Binomial distribution 
together with (\ref{eq:rate.better.sandwich}), implies
\beqn
\label{eq:sandwich.stopping.time}
\nonumber \pr \Big[\bin \Big(n-1,1 - (1 - p^{-})^t \Big) \geq t\Big] &\leq& \pr \Big[ \bin \Big(n-1,1 - \prod_{\tau=0}^{t-1} (1-r_\tau) \Big) \geq t\Big] \\
&\leq& \pr\Big[ \bin \Big(n-1,1 - (1 - p^{+})^t \Big) \geq t\Big].
\eeqn
%

%
%
In light of (\ref{eq:rate.better.sandwich}), the branching process $\{Y_t\}_{t \geq 0}$ associated to a RIG 
is stochastically bounded from below and form above 
by the branching processes associated to $G_{n,p^{-}}$ and $G_{n,p^{+}}$, respectively (for the analysis on an Erd\H{o}s-R\'enyi graph, see~\cite{alon-2000-probabilistic}).
Since both $c^{-}, c^{+}>1$, there exist giant components in both
$G_{n,p^{-}}$ and $G_{n,p^{+}}$, \textbf{whp}.  

In \cite{hofstad-notes}, it has been shown that the giant components in $G_{n,\lm/n}$, for $\lm>1$, is unique
and of size $\approx n \zl$, where $\zeta_\lambda$
is the unique solution from $(0,1)$ of the equation 
\beq
\label{eq.tree.size}
1-e^{-\lm \zeta} = \zeta.
\eeq 
Moreover,  the size of the giant component in $G_{n,\lm/n}$ satisfies the central
limit theorem 
\beq
\label{eq:clt.distr.}
\frac{\max_{v} \{|C(v)\}| - \zl n}{\sqrt{n}}  \stackrel{d}{=} \Ncal\Big(0,\frac{\zl (1-\zl)}{(1-\lm + \lm \zl)^2}\Big).
\eeq
From the definition of the stopping time, see (\ref{eq:prob.stopping.time}), and since (\ref{eq:sandwich.stopping.time}) and (\ref{eq:clt.distr.}), it follows  there is a giant component in a RIG, of size, at least, $n \zl (1 - \o(1))$, \whp. 
Furthermore, the stopping times of the branching processes associated
to $G_{n,p^{-}}$ and $G_{n,p^{+}}$ are approximately $\zeta n$, where $\zeta$ satisfy (\ref{eq.tree.size}), with $\lm^{-} = n p^{-}$ and $\lm^{+} = n p^{+}$, respectively. 
These two stopping times are close to one another, which follows from analyzing the
function $F(\zeta,c) = 1-\zeta - e^{-c\zeta}$, where $(\zeta,c)$ is the solution of 
$F(\zeta,c) = 0$, for given $c$.  Since all partial derivatives of $F(\zeta,c)$ are continuous and bounded, the stopping times of the branching processes defined from 
$G_{n,p^{-}}$, $G_{n,p^{+}}$ are `close' to the solution of (\ref{eq.tree.size}), for $\lm=c$.
From (\ref{eq:sandwich.stopping.time}), the stopping time of a RIG is bounded by the stopping times on $G_{n,p^{-}}$, $G_{n,p^{+}}$.

We conclude by proving that \whp,
the giant component of a RIG is unique by adapting the arguments 
in~\cite{alon-2000-probabilistic} to our setting.
Let us assume that there are at least two giant components in a RIG, with 
the sets of nodes $V_1, V_2 \subset V$. 
Let us create a new, independent `sprinkling'  $\widehat{\textrm{RIG}}$ on the top of our RIG, 
with the same sets of nodes and attributes, while $\hat{p}_w = p_w^\gamma$, for $\gamma>1$ to be defined later. 
Now, our object of interest is $\textrm{RIG}_{new} = \textrm{RIG} \cup \widehat{\textrm{RIG}}$. 
Let us consider all  $\Theta(n^2)$  pairs  $\{v_1,v_2\}$, where $v_1 \in V_1, v_2 \in V_2$, which are independent in $\widehat{\textrm{RIG}}$,
(but not in RIG), hence 
%
the probability that two nodes $v_1, v_2 \in V$ are connected in $\widehat{\textrm{RIG}}$ is given by 
\beq
\label{eq:edge.prob.in.righat}
1 - \prod_w (1 - \hat{p}_w^2) = 1 - \prod_w (1 - p_w^{2\gamma}) =  \sum_w p_w^{2\gamma} + \o(\sum_w p_w^{2\gamma}),
\eeq
which is true, since $\gamma>1$ and $p_w = O(1/n)$ for any $w$. Given that $\sum_w p_w^2 = c/n$, we choose $\gamma>1$ so that $\sum_w p_w^{2\gamma} = \omega(1/n^2)$.
Now, by the Markov inequality, \whp \mbox{}  there is a pair $\{v_1,v_2\}$ such that $v_1$ is connected to $v_2$ in $\widehat{\textrm{RIG}}$, 
implying that $V_1, V_2$ are connected, \textbf{whp}, forming one connected component within $\textrm{RIG}_{new}$. 
From the previous analysis, it follows that this component is of size at least $2n \zl (1 - \dl)$ for any small constant $\dl>0$. 
On the other hand, the probabilities $p_w^{new}$ in $\textrm{RIG}_{new}$ satisfy 
\beq
\nonumber
p_w^{new}= 1 - (1-p_w)(1-\hat{p}_w) = p_w + \hat{p}_w(1-p_w) = p_w + p_w^{\gamma}(1-p_w) = p_w (1 + \o(1)),
\eeq
which is again true, since $\gamma>1$ and $p_w = O(1/n)$ for any $w$. Thus, 
\beq
\sum_{w \in W} (p^{new}_w)^2 = \sum_{w \in W}  p_w^2 + \Theta(\sum_{w \in W} p_w^{1+\gamma} (1 - p_w)) =  \sum_{w \in W} p_w^2(1+\o(1)) = c/n + o(1/n).
\eeq 
Given that the stopping time on RIG is bounded by the stopping times on $G_{n,p^{-}}$, $G_{n,p^{+}}$, and from its  continuity, it follows that the giant component in $\textrm{RIG}_{new}$ cannot be of size $2n \zl (1 - \dl)$, which is a contradiction. 
Thus, there is only one giant component in RIG, of size given by $n \zeta_c (1 + \o(1))$, where $\zeta_c$ satisfies (\ref{eq.tree.size}), for $\lm = c$. Moreover, knowing behavior of $G_{n,p}$, from (\ref{eq:sandwich.stopping.time}), it follows that all other components are of size $O(\log n)$.
\end{proof}

\section{Conclusion}
\label{sec.conclusion}

The analysis of random models for bipartite graphs is important for the
study of social networks, or any network formed by associating nodes with
shared attributes.  In the random intersection graph (RIG) model, nodes have certain attributes with
fixed probabilities.  
In this paper, we have considered the general RIG model, where these probabilities are
represented by 
a set of probabilities $\boldsymbol{p} = \{p_w\}_{w \in W}$, where $p_w$
denotes the probability that a node is attached to the attribute $w$.

We have analyzed the evolution of components in general RIGs, giving
conditions for existence and uniqueness of the giant component.  We have
done so by generalizing the branching process argument used to study the
birth of the giant component in Erd\H{o}s-R\'enyi graphs.  We have
considered a dependent, inhomogeneous Galton-Watson process, where the
number of offspring follows a binomial distribution with a different
number of nodes and different rate at each step during the evolution.
The analysis of such a process is complicated by the dependence on its
history, dictated by the structure of general RIGs. 
%
We have shown that in spite of this difficulty, it is possible
to give stochastic bounds on the branching process, and that under
certain conditions the giant component appears at the
threshold $n \sum_{w \in W}  p_w^2 = 1$,  with probability tending to one, as the number of nodes tends to infinity.

\section*{Acknowledgments}
Part of this work was funded by the Department of Energy at Los
Alamos National Laboratory under contract DE-AC52-06NA25396 through the
Laboratory-Directed Research and Development Program, and by the National
Science Foundation grant CCF-0829945. Nicolas W. Hengartner was supported by DOE-LDRD 20080391ER.

\bibliographystyle{acm}
\bibliography{intergraph}
\end{document}